\documentclass[10pt,journal,compsoc]{IEEEtran}

\usepackage{cite}
\usepackage{amsmath,amssymb,amsfonts,amsthm}
\usepackage{algorithmic}
\usepackage{graphicx,multicol}
\usepackage{textcomp}
\def\BibTeX{{\rm B\kern-.05em{\sc i\kern-.025em b}\kern-.08em
    T\kern-.1667em\lower.7ex\hbox{E}\kern-.125emX}}
\usepackage{epstopdf}
\usepackage{caption}
\usepackage{url}
\usepackage{cite}
\usepackage{color}
\usepackage{subcaption}
\usepackage{marvosym}

\usepackage{tikz}
\usepackage{multirow}
\usetikzlibrary{shapes,arrows,calc}
\usetikzlibrary{positioning}
\usetikzlibrary{arrows,automata}

\usepackage{amsmath} 
\usepackage{dsfont}
\usepackage{standalone}
\usepackage{booktabs, caption, array}
\usepackage{siunitx}
\usepackage{epstopdf}
\usepackage{epsfig} 
\usepackage{psfrag}
\usepackage{lipsum}
\usepackage{amssymb}  
\usepackage{cite}
\usepackage{wrapfig}
\usepackage{mathrsfs}
\usepackage{url}
\usepackage{color}
\usepackage{colortbl}

\newtheorem{theorem}{Theorem}

\newtheorem{lemma}{Lemma}
\newtheorem{remark}{Remark}

\newtheorem{proposition}{Proposition}

\allowdisplaybreaks

\newcommand{\ignore}[1]{}

\newcommand{\Rb}{\mathbb{R}}

\newcommand{\ST}{\mathtt{S}}
\newcommand{\IT}{\mathtt{I}}
\newcommand{\GG}{\mathcal{G}}
\newcommand{\VV}{\mathcal{V}}
\newcommand{\EE}{\mathcal{E}}

\newcommand{\II}{\mathcal{I}}
\newcommand{\RR}{\mathcal{R}}

\newcommand{\DD}{\mathcal{D}}

\newcommand{\SUt}{\mathtt{SU}}
\newcommand{\SPt}{\mathtt{SP}}
\newcommand{\IPt}{\mathtt{IP}}
\newcommand{\IUt}{\mathtt{IU}}

\newcommand{\betau}{\beta_\mathtt{U}}
\newcommand{\betap}{\beta_\mathtt{P}}
\newcommand{\cp}{c_\mathtt{P}}

\newcommand{\oprocendsymbol}{\hbox{$\bullet$}}
\newcommand{\oprocend}{\relax\ifmmode\else\unskip\hfill\fi\oprocendsymbol}

\allowdisplaybreaks

\begin{document}

\title{Game-Theoretic Protection Adoption Against Networked SIS Epidemics}

\author{Abhisek~Satapathi and
        Ashish R. Hota,~\IEEEmembership{Senior Member, IEEE}% <-this % stops a space
\IEEEcompsocitemizethanks{\IEEEcompsocthanksitem A. Satapathi and A. R. Hota are with the Department of Electrical Engineering, Indian Institute of Technology (IIT) Kharagpur,
West Bengal, India, 721302.\protect\\
E-mail: abhisek.ee@iitkgp.ac.in, ahota@ee.iitkgp.ac.in}}

\IEEEtitleabstractindextext{%
\begin{abstract}
In this paper, we investigate game-theoretic strategies for containing spreading processes on large-scale networks. Specifically, we consider the class of networked susceptible-infected-susceptible (SIS) epidemics where a large population of agents strategically choose whether to adopt partially effective protection. We define the utilities of the agents which depends on the degree of the agent, its individual infection status and action, as well as the the overall prevalence of the epidemic and strategy profile of the entire population. We further present the coupled dynamics of epidemic evolution as well as strategy update which is assumed to follow the replicator dynamics. By relying on timescale separation arguments, we first derive the optimal strategy of protection adoption by the agents for a given epidemic state, and then present the reduced epidemic dynamics. The existence and uniqueness of endemic equilibrium is rigorously characterized and forms the main result of this paper. Finally, we present extensive numerical results to highlight the impacts of heterogeneous node degrees, infection rates, cost of protection adoption, and effectiveness of protection on the epidemic prevalence at the equilibrium.   
\end{abstract}

}

%Manuscript Types: Regular Paper Submission: Area 4 - Emerging Networks (e.g., energy networks, social networks, economic networks, transportation networks, and biological networks)
% Keywords: 11.5 Physiological networks and epidemiological networks < 11. Biological and ecological networks, 8.2 Opinion dynamics < 8. Social networks, 2.4 Network dynamics < 2. Network Science, 1.6 Network optimization and control < 1. General network topics

\maketitle

\section{Introduction}
\label{section:introduction}

Effective containment of spreading processes, such as infectious diseases spreading in society \cite{nowzari2016analysis,hethcote2000mathematics}, opinions spreading via social networks \cite{lopez2008diffusion} and viruses spreading on computer networks \cite{sellke2008modeling}, has proven to be challenging for two main reasons. First, deploying centralized control strategies is often impractical due to the large-scale nature of the networked system, and hence, decentralized strategies need to be developed \cite{huang2022game}. Second, the entities present in the network are often heterogeneous in terms of their connectivity patterns \cite{newman2010networks,zino2021analysis}. In order to address the above challenges, we propose a game-theoretic model where a large number of heterogeneous agents strategically choose to adopt (partially effective) protection measures against the class of Susceptible-Infected-Susceptible (SIS) epidemics on networks. 

While mathematical modeling and analysis of epidemics has a long history \cite{hethcote2000mathematics,nowzari2016analysis,pastor2015epidemic}, the recent COVID-19 pandemic has led to renewed interest in this topic. In particular, a substantial body of recent work has examined the impacts of decentralized containment strategies against epidemics in the framework of game theory; see \cite{huang2022game} for a recent comprehensive review. Two broad classes of containment strategies have been examined using game theory. In the first line of work, the decision-makers or agents adopt vaccination against the disease \cite{bauch2004vaccination,trajanovski2015decentralized,trajanovski2017designing,hota2019game} by evaluating the trade-off between the cost of vaccine and probability of being infected in the steady-state of the SIS epidemic dynamics. While earlier works (such as \cite{bauch2004vaccination}) considered a homogeneous population of agents, later works (such as \cite{trajanovski2015decentralized,hota2019game,frieswijk2023polarized}) considered agents with heterogeneous degrees. 

In the second line of work, agents adopt protection measures, such as wearing masks, social distancing, etc., against the epidemic. While vaccination is typically a one-time irreversible decision, protection adoption is often reversible, and the agents can revise their action or strategy dynamically as the epidemic builds up or wanes. As a result, the evolution of the epidemic and the actions of the agents often evolve in a comparable time-scale. Accordingly, several recent papers, including \cite{theodorakopoulos2013selfish,paarporn2023sis,frieswijk2022mean,satapathi2023coupled}, have analyzed the dynamics of coupled evolution of the SIS epidemic and protection adoption behavior, its equilibria and stability for a large population of homogeneous agents, while \cite{eksin2016disease,hota2022impacts,liu2022herd} have studied the above phenomenon for agents on a network with heterogeneous node degrees. A few related works \cite{she2022networked,xu2024discrete} have also examined the coupled evolution of opinion and epidemic as well as opinion and action \cite{aghbolagh2023coevolutionary}. Finally, similar investigations \cite{altman2022mask,amini2022epidemic,kordonis2022dynamic} have also been carried out for the class of Susceptible-Infected-Recovered (SIR) epidemics and its variants where recovery grants permanent immunity from future infections. 

In this paper, we generalize the settings examined in prior works \cite{satapathi2022epidemic,satapathi2023coupled} to include networked interaction among agents, and further generalize the setting in \cite{hota2022impacts} to account for partially effective protection.  Specifically, each agent is either susceptible or infected at a given point of time, and chooses whether to adopt protection or remain unprotected. For a susceptible individual, adopting protection gives partial immunity from the disease, while an infected protected individual causes new infection with a smaller probability compared to an infected individual who does not adopt protection. We assume that the agents are divided into different sub-populations depending on their degree, and that the epidemic evolution is governed by the degree-based mean-field (DBMF) approximation of the SIS epidemic \cite{pastor2015epidemic}. 

For the proposed setting, we define the utility of each sub-population which depends on their individual infection status, chosen action and the disease state and strategies adopted by the entire population. We formulate the coupled disease-behavior dynamics, and leverage time-scale separation arguments to derive the Nash equilibrium strategies for a given epidemic state, i.e., we assume that the evolution of protection adoption is faster than the evolution of the disease. Thus, our work is complementary to the closely related setting \cite{liu2022herd} which assumed the epidemic dynamics to be the faster dynamics and the behavior adoption to be the slower dynamics. We then derive the (slower) dynamics of epidemic evolution when all agents adopt their equilibrium strategies; this dynamics takes the form of a switched or hybrid system. We rigorously prove the existence and uniqueness of its equilibrium by leveraging several structural properties of the endemic equilibrium. Numerical simulations show that the coupled dynamics converges to the equilibrium. In addition, we numerically illustrate the impacts of heterogeneous node degrees, degree-dependent infection rates, and cost of protection adoption on the expected fraction of infected nodes at the equilibrium.  

%%%%%%%%%%%%%%%%%%%%%%%%%%%%%%%%%%%%%%%%%%%%%%%%%%%%%%%%%%%%%%%%%%%%%%%%%%%
%%%%%%%%%%%%%%%%%%%%%%%%%%%%%%%%%%%%%%%%%%%%%%%%%%%%%%%%%%%%%%%%%%%%%%%%%%%
%%%%%%%%%%%%%%%%%%%%%%%%%%%%%%%%%%%%%%%%%%%%%%%%%%%%%%%%%%%%%%%%%%%%%%%%%%%

\section{Networked SIS Epidemic under Partially Effective Protection}
\label{section:dbmf_game}

\subsection{Degree-Based Mean-Field Approximation}
We consider a large population of agents, where each agent has a specific degree (number of neighbors) from the set $\DD \in \{1,2,\ldots,d^{\max}\}$. Let $y^d(t) \in [0,1]$ be the proportion of agents with degree $d \in \DD$ that is infected at time $t$, with $1-y^d(t)$ be the proportion that is susceptible. Let $z^d_{\ST}(t)$ and $z^d_{\IT}(t)$ denote the proportion of susceptible and infected agents with degree $d$ that remain unprotected at time $t$, respectively. Let $\mathbf{z} := \{z^d_{\ST},z^d_{\IT}\}_{d\in\DD}$ the strategy profile of the entire population, and let $\mathbf{x} := \{y^d,z^d_{\ST},z^d_{\IT}\}_{d\in\DD}$ denote the (time-varying) social state. 

An infected individual with degree $d$ transmits the infection with probability $\betap^d \in (0,1)$ when it adopts protection, and with probability $\betau^d \in (0,1)$ when it is unprotected. A susceptible individual that adopts protection is $\alpha \in (0,1)$ times (less) likely to become infected compared to a susceptible unprotected individual. Finally, an infected individual recovers with probability $\gamma \in (0,1)$. 

Following the DBMF approximation of the SIS epidemic model, the infected proportion $y^d(t)$ evolves in continuous-time as
\begin{equation}\label{eq:sis_degree_d}
    \dot{y}^d(t) = -\gamma y^d(t) + (1-y^d(t)) (z^d_{\ST}(t) + \alpha (1-z^d_{\ST}(t))) d \Theta(\mathbf{x}(t)),
\end{equation}
where $\Theta(\mathbf{x})$ is the probability with which a randomly chosen neighbor of a node with degree $d$ transmits infection to it. The quantity $(z^d_{\ST}(t) + \alpha (1-z^d_{\ST}(t)))$ captures the fact that among the susceptible proportion, $z^d_{\ST}(t)$ fraction does not adopt protection and encounters an infection probability given by $d \Theta(\mathbf{x}(t))$, while $1-z^d_{\ST}(t)$ fraction adopts protection and encounters a smaller infection probability $\alpha d \Theta(\mathbf{x}(t))$. We define
\begin{equation}\label{eq:def_theta}
\Theta(\mathbf{x}) := \sum_{d \in \DD} \Big[\frac{dm_d}{d^{\mathtt{avg}}}  (\betau^dz^d_{\IT}+\betap^d(1-z^d_{\IT})) y^d \Big],
\end{equation}
where $m_d$ is the proportion of nodes with degree $d$ in the entire population and $d^{\mathtt{avg}} = \sum_{d \in \DD} dm_d$. Without loss of generality, we assume $m_d > 0$ for all $d \in \DD$. The first term specifies the probability of a randomly chosen neighbor having degree $d$ in accordance with the configuration model \cite{newman2010networks}. The second term denotes the probability of becoming infected if it comes in contact with an infected neighbor of degree $d$ which depends on the strategy adopted by infected agents having degree $d$. The third term denotes the probability of the neighbor with degree $d$ being infected in the first place. 

We now analyze the steady-state of the epidemic dynamics \eqref{eq:sis_degree_d} for a given strategy profile $\mathbf{z}$ of the population. By setting $\dot{y}^d(t)=0$, we obtain
\begin{align}
    & \gamma y^d = (1-y^d) (z^d_{\ST} + \alpha (1-z^d_{\ST})) d \Theta(\mathbf{z})) \nonumber
    \\ \implies & (\gamma + (z^d_{\ST} + \alpha (1-z^d_{\ST})) d \Theta(\mathbf{z})) y^d \nonumber
    \\ & \qquad = (z^d_{\ST} + \alpha (1-z^d_{\ST})) d \Theta(\mathbf{z}) 
    \\ \implies & y^d = \frac{(z^d_{\ST} + \alpha (1-z^d_{\ST})) d \Theta(\mathbf{z})}{\gamma + (z^d_{\ST} + \alpha (1-z^d_{\ST})) d \Theta(\mathbf{z})} \label{eq:yd_equilibrium}
    \\ \implies & \Theta(\mathbf{z}) = \sum_{d \in \DD} \Big[\frac{dm_d}{d^{\mathtt{avg}}} (\betau^dz^d_{\IT}+\betap^d(1-z^d_{\IT})) \times \nonumber
    \\ & \qquad \qquad \qquad \qquad \frac{(z^d_{\ST} + \alpha (1-z^d_{\ST})) d \Theta(\mathbf{z})}{\gamma + (z^d_{\ST} + \alpha (1-z^d_{\ST})) d \Theta(\mathbf{z})} \Big] \label{eq:theta_equilibrium_1}
    \\ \implies & \Theta(\mathbf{z}) \Big[1- \sum_{d \in \DD} \Big[\frac{dm_d}{d^{\mathtt{avg}}} (\betau^dz^d_{\IT}+\betap^d(1-z^d_{\IT})) \times \nonumber
    \\ & \qquad \qquad \frac{(z^d_{\ST} + \alpha (1-z^d_{\ST})) d }{\gamma + (z^d_{\ST} + \alpha (1-z^d_{\ST})) d \Theta(\mathbf{z})}\Big]\Big] = 0. \label{eq:theta_equilibrium_identity}
\end{align}
Note that \eqref{eq:theta_equilibrium_1} is obtained by substituting the expression of $y^d$ obtained in \eqref{eq:yd_equilibrium} in the definition of $\Theta(\mathbf{z})$ given in \eqref{eq:def_theta}. It now follows that $\Theta(\mathbf{z})=0$ is always a solution of \eqref{eq:theta_equilibrium_identity} which corresponds to the disease-free equilibrium. In addition, there may be nonzero solution(s) $\Theta(\mathbf{z})$ of \eqref{eq:theta_equilibrium_identity} depending on the strategy profile $\mathbf{z}$ and other parameters as formalized in the following lemma.

\begin{lemma}\label{lemma:dbmf_theta_nonzero}
Equation \eqref{eq:theta_equilibrium_identity} admits a nonzero solution $\Theta^\star(\mathbf{z})$ if and only if 
    \begin{equation*}
        1 < \frac{1}{\gamma}\sum_{d \in \DD} \Big[\frac{d^2m_d}{d^{\mathtt{avg}}} (\betau^dz^d_{\IT}+\betap^d(1-z^d_{\IT}))  (z^d_{\ST} + \alpha (1-z^d_{\ST}))\Big].
    \end{equation*}
Furthermore, $\Theta^\star(\mathbf{z})=1$ is not a solution of \eqref{eq:theta_equilibrium_identity}.
\end{lemma}
\begin{proof}
    Note that for $\Theta^\star(\mathbf{z}) > 0$ to be a solution of \eqref{eq:theta_equilibrium_identity}, we must have
    \begin{align*}
        1 & = \sum_{d \in \DD} \Big[\frac{dm_d}{d^{\mathtt{avg}}} (\betau^dz^d_{\IT}+\betap^d(1-z^d_{\IT})) \times
        \\ & \qquad \qquad \frac{(z^d_{\ST} + \alpha (1-z^d_{\ST})) d }{\gamma + (z^d_{\ST} + \alpha (1-z^d_{\ST})) d \Theta^\star(\mathbf{z})}\Big].
    \end{align*}
    Note that the R.H.S. is monotonically decreasing in $\Theta^\star(\mathbf{z})$.     
    At $\Theta^\star(\mathbf{z}) = 1$, we have
    \begin{align*}
        & \sum_{d \in \DD} \Big[\frac{dm_d}{d^{\mathtt{avg}}} (\betau^dz^d_{\IT}+\betap^d(1-z^d_{\IT}))  \frac{(z^d_{\ST} + \alpha (1-z^d_{\ST})) d }{\gamma + (z^d_{\ST} + \alpha (1-z^d_{\ST})) d}\Big] 
        \\ < & \sum_{d \in \DD} \Big[\frac{dm_d}{d^{\mathtt{avg}}} (\betau^dz^d_{\IT}+\betap^d(1-z^d_{\IT})) \Big] 
        \\ < & \sum_{d \in \DD} \frac{dm_d}{d^{\mathtt{avg}}} = 1.
    \end{align*}
    Therefore, $\Theta^\star(\mathbf{z}) = 1$ is not a solution of \eqref{eq:theta_equilibrium_identity}. 
    Therefore, a nonzero solution exists if and only if the R.H.S. exceeds $1$ at $\Theta^\star(\mathbf{z}) = 0$. In other words, 
    \begin{equation*}
        1 < \frac{1}{\gamma}\sum_{d \in \DD} \Big[\frac{d^2m_d}{d^{\mathtt{avg}}} (\betau^dz^d_{\IT}+\betap^d(1-z^d_{\IT}))  (z^d_{\ST} + \alpha (1-z^d_{\ST}))\Big].
    \end{equation*}
    This concludes the proof.
\end{proof}

In the following subsection, we analyze the equilibrium behavior when agents adopt protection in a game-theoretic manner.

%%%%%%%%%%%%%%%%%%%%%%%%%%%%%%%%%%%%%%%%%%%%%%%%%%%%%%%%%%%%%%%%%%

\subsection{Game-Theoretic Protection Adoption}

We now define the payoff vector of the agents. A susceptible agent aims to balance the trade-off between the cost of adopting protection, denoted by the parameter $\cp > 0$, and the expected loss upon becoming infected. The expected loss is computed as the product of loss upon infection, captured by a parameter $L>0$, and the instantaneous probability of becoming infected. The latter quantity depends on the current social state $\mathbf{x}$, degree of the agent, and the action chosen by the agent. Formally, for a susceptible agent with degree $d$, we define its payoffs to be 
\begin{equation}\label{eq:payoff_susceptible}
    F^d_{\SUt}(\mathbf{x}) = - L d \Theta(\mathbf{x}), \qquad F^d_{\SPt}(\mathbf{x}) = - \cp - L \alpha d \Theta(\mathbf{x}),
\end{equation}
if the agent remains unprotected and adopts protection, respectively. In particular, when the agent adopts protection, it encounters an additional cost $\cp$ though its probability of becoming infected reduces due to the multiplying factor $\alpha \in (0,1)$. 

In contrast, an infected agent is already infected, and there is no immediate risk of becoming infected. Consequently, we define 
\begin{equation}\label{eq:payoff_infected}
    F^d_{\IUt}(\mathbf{x}) = - c_{\IUt}, \qquad F^d_{\IPt}(\mathbf{x}) = - c_{\IPt},
\end{equation}
to be the payoff of an infected agent if it remains unprotected and adopts protection, respectively. The parameter $c_{\IUt} > 0$ captures the penalty imposed on an infected agent if it does not adopt protection (or adhere to quarantine norms), while $c_{\IPt} > 0$ captures the inconvenience caused due to adopting protection while being sick. Note that the payoffs do not depend on the degree $d$ or social state $\mathbf{x}$. We further assume $c_{\IUt} > c_{\IPt}$, which indicates that infected agents prefer to adopt protection. 

%%%%%%%%%%%%%%%%%%%%%%%%%%%%%%%%%%%%%%%%%%%%%%

We assume that agents revise their protection adoption strategies following the replicator dynamics~\cite{sandholm2010population}. We further assume that agents only replicate the strategies of other agents who have the same degree and infection status, i.e., susceptible individuals with degree $d$ only replicate the strategies of other susceptible individuals of the same degree $d$ (likewise for infected individuals). Consequently, the proportion of unprotected susceptible nodes of degree $d$ evolves as
\begin{align}
\dot{z}^d_{\mathtt{S}}(t) & =  z^d_{\mathtt{S}}(t)(1-z^d_{\mathtt{S}}(t)) \big[ F^d_{\SUt}(\mathbf{x}(t)) - F^d_{\SPt}(\mathbf{x}(t)) \big] \nonumber
\\ & = z^d_{\mathtt{S}}(t)(1-z^d_{\mathtt{S}}(t)) \big[ \cp - L(1-\alpha) d \Theta(\mathbf{x}(t)) \big]. \label{eq:tnse_zs_replicator}
\end{align}
Similarly, for infected individuals, we have
\begin{align}
\dot{z}^d_{\mathtt{I}}(t) & =  z^d_{\mathtt{I}}(t)(1-z^d_{\mathtt{I}}(t)) (c_{\IPt}-c_{\IUt}). \label{eq:tnse_zi_replicator}
\end{align}

Thus, equations \eqref{eq:sis_degree_d}, \eqref{eq:tnse_zs_replicator} and \eqref{eq:tnse_zi_replicator} characterize the coupled evolution of the epidemic and population states at the same time-scale. The above set of dynamics has a large number of equilibrium points due to the structure of the replicator dynamics which induces stationary points at both $0$ and $1$. In order to obtain further insights into the behavior of the above system, and characterize the equilibrium points, we analyze the coupled dynamics under timescale separation. Motivated by the past work \cite{satapathi2023coupled}, we focus on the case in which the replicator dynamics evolves faster than the epidemic dynamics. This is a reasonable assumption as agents are likely to adjust their behavior faster than the spread of the epidemic due to increased awareness derived from conventional as well as social media. The coupled dynamics can now be written as a slow-fast system given by
\begin{align}\label{eq:tnse_coupled-dynamics-timescale_slow_epi}
\begin{split}
\dot{y}^d(t) & = -\!\gamma y^d(t) \!+\! (1\!-y^d(t)) (z^d_{\ST}(t) \!+ \!\alpha (1\!-\!z^d_{\ST}(t))) d \Theta(\mathbf{x}(t)), \\
\epsilon \dot{z}^d_{\mathtt{S}}(t) & = z^d_{\mathtt{S}}(t)(1-z^d_{\mathtt{S}}(t)) \big[ \cp - L(1-\alpha) d \Theta(\mathbf{x}(t)) \big], \\
\epsilon \dot{z}^d_{\mathtt{I}}(t) & =  z^d_{\mathtt{I}}(t)(1-z^d_{\mathtt{I}}(t)) (c_{\IPt}-c_{\IUt}),
\end{split}
\end{align}
for all $d \in \DD$, and where $\epsilon \in (0,1)$ is a timescale separation variable \cite{berglund2006noise}. 

%%%%%%%%%%%%%%%%%%%%%%%%%%%%%%%%%%%%%%%%%%%%%%%%%%%%%%

We first characterize the behavior of the agents (captured by \eqref{eq:tnse_zs_replicator} and \eqref{eq:tnse_zi_replicator}) for a given infection state $\mathbf{y} = \{y^d\}_{d \in \DD}$. It follows from \eqref{eq:payoff_infected} and our assumption $c_{\IUt} > c_{\IPt}$ that infected agents strictly prefer to adopt protection irrespective of the social state, and as a result, $z^d_{\IT} = 0$ is the unique stable equilibrium point of \eqref{eq:tnse_zi_replicator}. 

From the R.H.S. of \eqref{eq:def_theta}, it follows that when $z^d_{\IT} = 0$, $\Theta(\mathbf{x})$ does not depend on $z^d_{\ST}$ when $\mathbf{y}$ is specified. Therefore, with a slight abuse of notation, we define
\begin{equation}\label{eq:def_theta_2}
\Theta(\mathbf{y}) := \sum_{d \in \DD} \Big[\frac{dm_d}{d^{\mathtt{avg}}}  \betap^d y^d \Big].
\end{equation}

Now, for a susceptible agent, adopting protection is strictly preferred if and only if
\begin{align}
    F^d_{\SUt}(\mathbf{x}) < F^d_{\SPt}(\mathbf{x}) \iff & - L d \Theta(\mathbf{y}) < - \cp - L \alpha d \Theta(\mathbf{y}) \nonumber
    \\ \iff & \cp < L(1-\alpha)d \Theta(\mathbf{y}) \nonumber 
    \\ \iff & \Theta(\mathbf{y}) > \frac{\cp}{L(1-\alpha)d} =: \Theta^d_{th}. \label{eq:payoff_susceptible_difference}
\end{align}
In other words, the optimal strategy of a susceptible agent with degree $d$ depends on whether $\Theta(\mathbf{y})$ exceeds the degree-specific threshold $\Theta^d_{th}$ defined in \eqref{eq:payoff_susceptible_difference}. If $\Theta(\mathbf{y}) > \Theta^d_{th}$, then $z^d_{\ST} = 0$ is the stable equilibrium of \eqref{eq:tnse_zs_replicator}, while if $\Theta(\mathbf{y}) < \Theta^d_{th}$, then the $z^d_{\ST} = 1$ is the stable equilibrium of \eqref{eq:tnse_zs_replicator}. If $\Theta(\mathbf{y}) = \Theta^d_{th}$, then any $z^d_{\ST} \in [0,1]$ could emerge as the equilibrium of \eqref{eq:tnse_zs_replicator}. Thus, at a given $\mathbf{y}$, the replicator dynamics (which is the fast system since $\epsilon<1$) associated with both infected and susceptible agents has a unique stable equilibrium point except at a point with measure zero.  

We now state the reduced dynamics for the epidemic (which is the slow system) as 
\begin{subequations}\label{eq:sis_switched}
\begin{align}
\Theta(\mathbf{y}(t)) < \Theta^d_{th}: & \quad \dot{y}^d(t) = -\gamma y^d(t) + (1-y^d(t)) d \Theta(\mathbf{y}(t)), \label{eq:sis_lowy}
\\ \Theta(\mathbf{y}(t)) = \Theta^d_{th}: & \quad \dot{y}^d(t) \in \{  -\gamma y^d(t) + (1-y^d(t)) \times \nonumber
\\ & \quad (z^d_{\ST} + \alpha (1-z^d_{\ST})) d \Theta(\mathbf{y}(t)) \ \rvert \ z^d_{\ST} \in [0,1] \},  
\label{eq:sis_th1}
\\ \Theta(\mathbf{y}(t)) > \Theta^d_{th}: & \quad \dot{y}^d(t) = -\gamma y^d(t) + (1-y^d(t)) \alpha d \Theta(\mathbf{y}(t)). \label{eq:sis_highy}
\end{align}
\end{subequations}
The above dynamics approximates the coupled dynamics \eqref{eq:tnse_coupled-dynamics-timescale_slow_epi} as $\epsilon \to 0$, i.e., when individuals adopt protection in a strategic manner to maximize their payoffs as a function of the current epidemic state. Note that \eqref{eq:sis_switched} admits a Filippov solution \cite[Proposition 3]{cortes2008discontinuous} solution because the R.H.S. of \eqref{eq:sis_th1} is measurable and bounded. 

\begin{remark}
The dynamical system \eqref{eq:sis_switched} can be viewed as the dynamics of epidemic evolution when central authorities restrict interaction of nodes of a certain degree $d$ whenever $\Theta(\mathbf{y}(t))$ exceeds the threshold $\Theta^d_{th}$.  
\end{remark}

%%%%%%%%%%%%%%%%%%%%%%%%%%%%%%%%%%%%%%%%%%%%%%%%%
%%%%%%%%%%%%%%%%%%%%%%%%%%%%%%%%%%%%%%%%%%%%%%%%%
%%%%%%%%%%%%%%%%%%%%%%%%%%%%%%%%%%%%%%%%%%%%%%%%%
\section{Analysis of Equilibrium}

In this section, we characterize the existence and uniqueness of equilibrium of the dynamics \eqref{eq:sis_switched}. First observe that the thresholds $\Theta^d_{th}$, defined in \eqref{eq:payoff_susceptible_difference}, are monotonically decreasing in the degree $d$, i.e., agents with a larger degree switch to adopting protection for a smaller value of $\Theta(\mathbf{y}(t))$. Let $d_{\min}$ be the smallest degree for which $\Theta^{d_{\min}}_{th} < 1$. Then, we divide the region $[0,1]$ into $(d^{\max}-d_{\min}+2)$ number of intervals, denoted $\{\II_{d^{\max}+1},\II_{d^{\max}},\ldots,\II_{d_{\min}}\}$ such that
\begin{subequations}\label{eq:interval_II_def}
    \begin{align}
        & \II_{d^{\max}+1} := [0,\Theta^{d^{\max}}_{th}), 
        \\ & \II_{d_{\min}} := (\Theta^{d_{\min}}_{th},1], 
        \\ & \II_{d} := (\Theta^{d}_{th},\Theta^{d-1}_{th}), 
    \end{align}
\end{subequations}
for $d \in \{d^{\max},d^{\max}-1,\ldots,d_{\min}+1\}$.

We will now closely examine the dynamics \eqref{eq:sis_switched} when $\Theta(\mathbf{y}(t))$ belongs to one of the intervals as stated above. If $\Theta(\mathbf{y}(t)) \in \II_{d^\star}$, then $\Theta(\mathbf{y}(t)) > \Theta^d_{th}$ for all degree $d \geq d^\star$, and $\Theta(\mathbf{y}(t)) < \Theta^d_{th}$ for all degree $d < d^\star$. For this particular regime, the dynamics \eqref{eq:sis_switched} can be stated as:
\begin{subequations}\label{eq:sis_switched_2}
\begin{align}
d < d^\star: & \quad \dot{y}^d(t) = -\gamma y^d(t) + (1-y^d(t)) \times \nonumber 
\\ & \qquad \qquad \qquad d \sum_{d' \in \DD} \Big[\frac{d'm_{d'}}{d^{\mathtt{avg}}}  \betap^{d'} y^{d'}(t) \Big], \label{eq:sis2_lowy}
\\ d \geq d^\star: & \quad \dot{y}^d(t) = -\gamma y^d(t) + (1-y^d(t)) \times \nonumber 
\\ & \qquad \qquad \qquad \alpha d \sum_{d' \in \DD} \Big[\frac{d'm_{d'}}{d^{\mathtt{avg}}}  \betap^{d'} y^{d'}(t) \Big]. \label{eq:sis2_highy}
\end{align}
\end{subequations}

Before analyzing the equilibria of \eqref{eq:sis_switched}, we prove the following result on the equilibria of \eqref{eq:sis_switched_2}. We first define the following quantity: 
$$\RR(d^\star) := \sum^{d^\star-1}_{d=1} \frac{d^2m_d\betap^{d}}{d^{\mathtt{avg}} \gamma} + \sum^{d^{\max}}_{d = d^\star} \frac{\alpha d^2m_d\betap^{d}}{d^{\mathtt{avg}} \gamma}. $$ 

\begin{proposition}\label{proposition:tnse_reduced_dynamics}
Consider the dynamics \eqref{eq:sis_switched_2} for a specified $d^{\star} \in \{d_{\min},\ldots,d^{\max}+1\}$. We have the following characterization of its equilibria.
\begin{enumerate}
    \item If $\RR(d^{\star}) \leq 1$, then the disease-free equilibrium is the only equilibrium of \eqref{eq:sis_switched_2}. 
    \item If $\RR(d^{\star}) > 1$, then in addition to the disease-free equilibrium, there exists a unique nonzero endemic equilibrium of \eqref{eq:sis_switched_2}. 
\end{enumerate}
\end{proposition}
The proof is presented in Appendix \ref{section:tnse_appendix_proposition}, and leverages connection between the dynamics in \eqref{eq:sis_switched_2}, and the N-Intertwined mean-field approximation (NIMFA) of the networked SIS epidemic dynamics \cite{van2008virus,khanafer2016stability}.

%%%%%%%%%%%%%%%%%%%%%%%%%%%%%%%%%%%%%%%%%%%%%%%%%%%

When $\RR(d^\star) > 1$, we denote the endemic equilibrium with $\mathbf{y}_{\mathtt{EE}}(d^{\star})$. Following \eqref{eq:theta_equilibrium_identity}, the quantity $\Theta(\mathbf{y}_{\mathtt{EE}}(d^{\star}))$ at the endemic equilibrium is the unique value satisfying
\begin{align}
    1 & = \sum^{d^\star-1}_{d =1} \Big[\frac{dm_d}{d^{\mathtt{avg}}} \frac{d \betap^d}{\gamma + d \Theta(\mathbf{y}_{\mathtt{EE}}(d^{\star}))}\Big] \nonumber
    \\ & \qquad + \sum^{d^{\max}}_{d =d^\star} \Big[\frac{dm_d}{d^{\mathtt{avg}}} \frac{\alpha d \betap^d}{\gamma + \alpha d \Theta(\mathbf{y}_{\mathtt{EE}}(d^{\star}))}\Big]. \label{eq:theta_star_d_star}
\end{align}
If $\RR(d^\star) \leq 1$, we define $\Theta(\mathbf{y}_{\mathtt{EE}}(d^{\star}))=0$. The following lemma establishes monotonicity of $\RR(d^\star)$ and $\Theta(\mathbf{y}_{\mathtt{EE}}(d^{\star}))$.

\begin{lemma}\label{lemma:monotonicity_R_Theta_DBMF}
The quantities $\RR(d^\star)$ and $\Theta(\mathbf{y}_{\mathtt{EE}}(d^{\star}))$ are monotonically increasing in $d^\star$.
\end{lemma}
\begin{proof}
Note that as $d^\star$ increases, some entries move from the second summation to the first summation in the definition of $\RR(d^\star)$, and the terms in the first summation are larger because $\alpha \in (0,1)$. Similarly, if $d^\star$ increases and $\Theta(\mathbf{y}_{\mathtt{EE}})$ remains unchanged, then the R.H.S. of \eqref{eq:theta_star_d_star} increases. In order to achieve R.H.S. equal to $1$, $\Theta(\mathbf{y}_{\mathtt{EE}}(d^{\star}))$ needs to increase. 
\end{proof}

%%%%%%%%%%%%%%%%%%%%%%%%%%%%%%%%%%%%%%%%%%%%%%%%%%%
%%%%%%%%%%%%%%%%%%%%%%%%%%%%%%%%%%%%%%%%%%%%%%%%%%%

We are now ready to establish the existence and uniqueness of the equilibrium of the dynamics \eqref{eq:sis_switched} by leveraging the results established above.  

\begin{theorem}\label{theorem:tnse_main}
For the dynamics \eqref{eq:sis_switched}, we have the following characterization of its equilibria.
\begin{enumerate}[leftmargin=*]
    \item If $\RR(d^{\max}+1) \leq 1$, then the disease-free equilibrium is the only equilibrium of \eqref{eq:sis_switched}. 
    \item Now suppose $\RR(d^{\max}+1) > 1$. Let $\Theta^{d^{\max}+1}_{th}:=0$ for convenience. Let $d^{\mathtt{eq}} \in \{d_{\min},d_{\min}+1,\ldots,d^{\max},d^{\max}+1\}$ be the smallest degree for which $\Theta(\mathbf{y}_{\mathtt{EE}}(d^{\mathtt{eq}})) > \Theta^{d^\mathtt{eq}}_{th}$. Then, we have the following two cases.
    \begin{enumerate}
        \item If $\Theta(\mathbf{y}_{\mathtt{EE}}(d^{\mathtt{eq}})) \in \II_{d^{\mathtt{eq}}}$, then $\mathbf{y}_{\mathtt{EE}}(d^{\mathtt{eq}})$ is the unique endemic equilibrium of \eqref{eq:sis_switched}. 
        \item If $\Theta(\mathbf{y}_{\mathtt{EE}}(d^{\mathtt{eq}})) \geq \Theta^{d^\mathtt{eq}-1}_{th}$, then there exists a unique endemic equilibrium with  $\{y^d\}_{d \in \DD}$ satisfying \eqref{eq:yd_equilibrium} with $\Theta(\mathbf{y}) = \Theta^{d^{\mathtt{eq}}-1}_{th}$. 
    \end{enumerate}    
\end{enumerate}
\end{theorem}

The proof is presented in Appendix \ref{section:tnse_appendix_theorem_proof}, and it relies on the monotonicity properties established in Lemma \ref{lemma:monotonicity_R_Theta_DBMF} as well as Proposition \ref{proposition:tnse_reduced_dynamics}. Note that parameters such as $\cp$ and $L$ affect the outcome by influencing the thresholds $\Theta^d_{th}$, while $\alpha$ affects both the thresholds as well as $\Theta(\mathbf{y}_{\mathtt{EE}}(d))$. In Case 2(a) of the theorem, all susceptible agents with degree $d \geq d^{\mathtt{eq}}$ adopt protection and all susceptible agents with degree $d < d^{\mathtt{eq}}$ remain unprotected. In contrast, in Case 2(b), the proportion of susceptible agents with degree $d^{\mathtt{eq}}$ that adopt protection is strictly between $0$ and $1$. It is easy to see that the endemic equilibrium identified in Theorem \ref{theorem:tnse_main} together with the above protection adoption scheme constitutes an equilibrium of the coupled dynamics \eqref{eq:sis_degree_d}, \eqref{eq:tnse_zs_replicator} and \eqref{eq:tnse_zi_replicator}. 

%%%%%%%%%%%%%%%%%%%%%%%%%%%%%%%%%%%%%%%%%%%%%%%%%%%%%%%%%%%%%%%%%
%%%%%%%%%%%%%%%%%%%%%%%%%%%%%%%%%%%%%%%%%%%%%%%%%%%%%%%%%%%%%%%%%
%%%%%%%%%%%%%%%%%%%%%%%%%%%%%%%%%%%%%%%%%%%%%%%%%%%%%%%%%%%%%%%%%

\begin{figure*}[t]
\centering
  \includegraphics[width=0.33\linewidth]{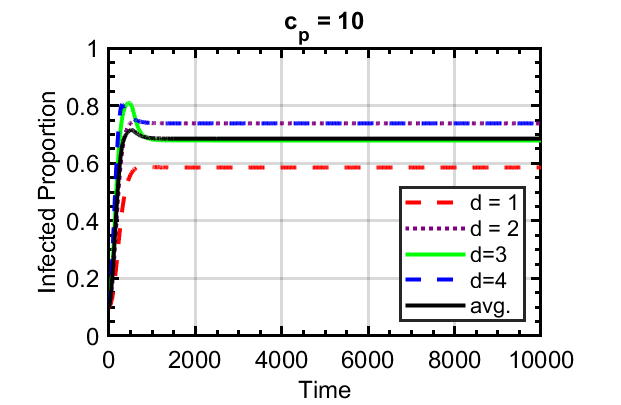}
  \includegraphics[width=0.33\linewidth]{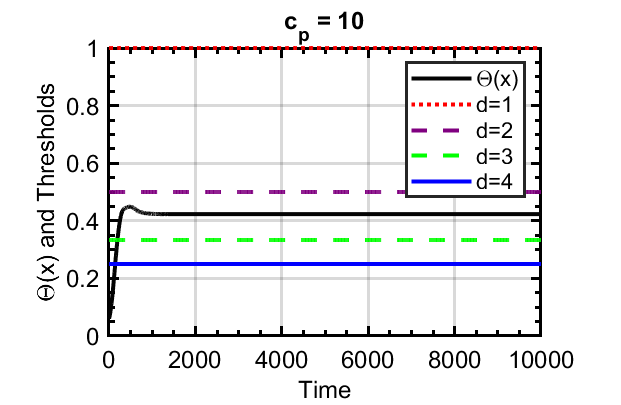}
  \includegraphics[width=0.33\linewidth]{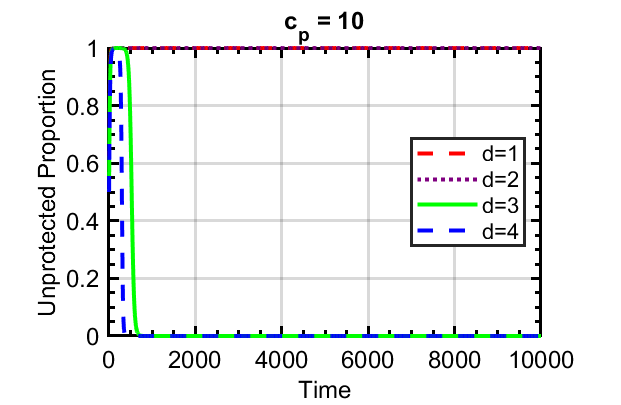}
  \\
  \includegraphics[width=0.33\linewidth]{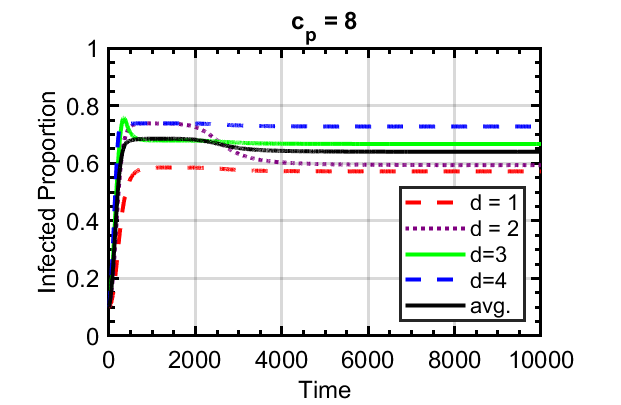}
  \includegraphics[width=0.33\linewidth]{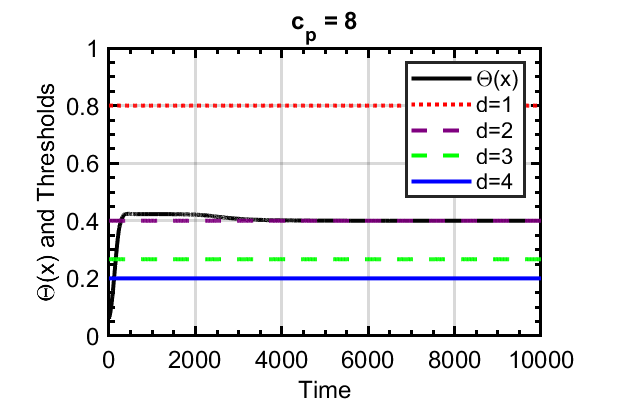}
  \includegraphics[width=0.33\linewidth]{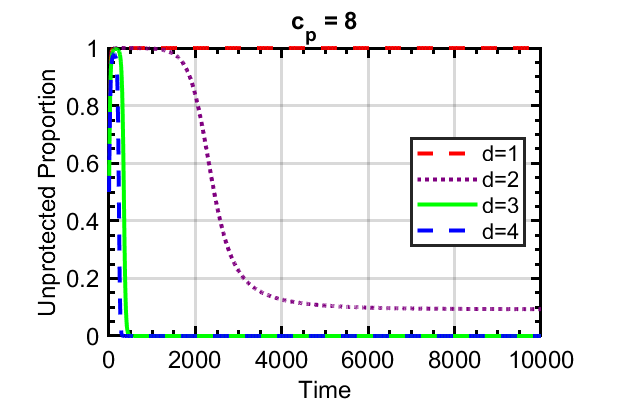}
  \caption{Evolution of infected proportion (left), $\Theta(\mathbf{x}(t)) = \Theta(\mathbf{y}(t))$ and degree-specific thresholds (middle), and proportion of unprotected susceptible agents (right) for $c_\mathtt{P} = 10$ (top row) and $c_\mathtt{P} = 8$ (bottom row).}
  \label{figure:tnse0}
\end{figure*}

\section{Numerical Results}

In this section, we numerically illustrate the convergence of the coupled dynamics and characteristics of the endemic equilibrium as a function of different parameters of the game, including cost of adopting protection, infection rate, effectiveness of protection and heterogeneous degree distributions.

\subsection{Convergence to Endemic Equilibrium}
\label{section:numerical_endemic} 

First we show that the coupled dynamics converges to the endemic equilibrium as established in Theorem \ref{theorem:tnse_main}. The values of different parameters used in this subsection are reported in Table \ref{tab:tnse0_parameter}. In particular, the values of $y^d(0)$ and $z^d_{\mathtt{S}}(0)$ denote the initial proportion of infected agents and unprotected susceptible agents. These initial conditions are used to compute the trajectories for all degrees using an Euler discretization of \eqref{eq:tnse_coupled-dynamics-timescale_slow_epi} with discretization parameter $0.01$ and $\epsilon=1$. The infection rates are identical for all degrees. We consider a network with the set of degrees $\DD = \{1,2,3,4\}$ with the proportion of nodes with each of the above degrees being $0.25$, i.e., $m_d = 0.25$ for all $d \in \DD$. The values of the thresholds $\Theta^d_{th}$ for two different values of $c_\mathtt{P}$ as well as the values $\Theta(\mathbf{y}_{\mathtt{EE}}(d))$ for $d \in \{2,3,4,5\}$ are reported in Table \ref{tab:tnse0_degree}.

\begin{table}[h]
\begin{center}
\begin{tabular}{|c | c | c | c | c | c | c | } 
 \hline
 $\alpha$ & $\beta^d_{\mathtt{P}}$ & $\beta^d_{\mathtt{U}}$ & $\gamma$ & $L$ & $y^d(0)$ & $z^d_{\mathtt{S}}(0)$ \\ [0.5ex] 
 \hline
 0.5 & 0.6 & 0.7 & 0.3 & 20 & 0.1 & 0.5 \\ \hline
\end{tabular}
 \caption{Values of different parameters used in simulations in Section \ref{section:numerical_endemic}.}
 \label{tab:tnse0_parameter}
\end{center}
 \end{table}

\begin{table}[h]
\begin{center}
\begin{tabular}{|c | c | c | c | c |} 
 \hline
 degree ($d$) & $m_d$ & $\Theta^d_{th}$ & $\Theta^d_{th}$ & $\Theta(\mathbf{y}_{\mathtt{EE}}(d+1))$ \\
 & & at $c_\mathtt{P}=10$ & at $c_\mathtt{P}=8$ & \\ 
 \hline
 1 & 0.25 & 1 & 0.8 & 0.3961 \\ \hline
 2 & 0.25 & 0.5 & 0.4 & 0.4231 \\ \hline
 3 & 0.25 & 0.33 & 0.2667 & 0.4543 \\ \hline
 4 & 0.25 & 0.25 & 0.2 & 0.4860 \\ \hline
\end{tabular}
 \caption{Degree distribution and thresholds used in simulations in Section \ref{section:numerical_endemic}.}
 \label{tab:tnse0_degree}
\end{center}
\end{table}

It follows from Table \ref{tab:tnse0_degree} that when $c_\mathtt{P}=10$, $d=3$ is the smallest degree for which $\Theta(\mathbf{y}_{\mathtt{EE}}(3)) = 0.4231 > 0.33 = \Theta^3_{th}$. Furthermore, $0.4231 \in \mathcal{I}_3 = (0.33,0.5)$. Consequently, the value of $\Theta(\mathbf{y})$ at the endemic equilibrium should be $0.4231$ following Theorem \ref{theorem:tnse_main}. This is precisely what we observe in the top row of Figure \ref{figure:tnse0}. 

The plots in the left panel of the top row of Figure \ref{figure:tnse0} show the evolution of infected proportion $y^d(t)$ for different degrees of the network as well as the expected fraction of infected nodes $y^{\mathtt{avg}}(t) = \sum_{d \in \DD} m_d y^d(t)$ shown in the thick black line. The infected proportions converge to the unique endemic equilibrium within $1500$ time steps. The plots in the middle panel show the values of thresholds $\Theta^d_{th}$ and $\Theta(\mathbf{y}(t))$, and indicate that $\Theta(\mathbf{y}(t))$ converges to the value $0.4231$ which lies in the interval $\mathcal{I}_3 = (0.33,0.5)$. The plots in the right panel show the evolution of the proportion unprotected susceptible agents $z^d_{\mathtt{S}}(t)$ for different degrees of the network. At the onset of the pandemic, when $\Theta(\mathbf{y}(t))$ was smaller than the thresholds for all the degrees, the quantity $z^d_{\mathtt{S}}(t)$ increased close to $1$ for all $d$. Eventually, as $\Theta(\mathbf{y}(t))$ exceeded the thresholds $\Theta^d_{th}$ for $d = 4$ and $d=3$, $z^4_{\mathtt{S}}(t)$ and $z^3_{\mathtt{S}}(t)$ started to decline and eventually converged to $0$ in accordance with the discussion in Section \ref{section:dbmf_game}.

We now examine the case where $c_\mathtt{P}=8$. According to Table \ref{tab:tnse0_degree}, $d=3$ is the smallest degree for which $\Theta(\mathbf{y}_{\mathtt{EE}}(3)) = 0.4231 > 0.2667 = \Theta^3_{th}$. However, in this case, $0.4231 \notin \mathcal{I}_3 = (0.2667,0.4)$. Consequently, the value of $\Theta(\mathbf{y})$ at the endemic equilibrium should be $0.4$ following Theorem \ref{theorem:tnse_main}. This is precisely what we observe in the bottom row of Figure \ref{figure:tnse0}. In particular, the plot in the middle panel shows that  $\Theta(\mathbf{y}(t))$ converges to $\Theta^2_{th}$. The plot in the right panel shows that $z^d_{\mathtt{S}}(t)$ converges to $0$ for $d = 3$ and $d=4$, while it converges to $1$ for $d=1$. However, $z^2_{\mathtt{S}}(t)$ converges to an intermediate value. Thus, the social state is shown to converge to the unique endemic equilibrium as postulated in Theorem \ref{theorem:tnse_main} in the both the cases. The convergence is slower in this case compared the case with $c_\mathtt{P}=10$. 

\subsection{Heterogeneous Degree Distribution and Infection Rates}
\label{section:tnse_casestudy1}

In the previous subsection, the infection rate $\beta^d_{\mathtt{P}}$ and the proportion $m_d$ were identical for all degrees. We now illustrate the impact of heterogeneity in these parameters. Let $\DD = \{1,2,3,4\}$ as before. We consider the following two cases.
\begin{itemize}
    \item Case 1: $\beta^d_{\mathtt{P}} = 0.1$ for $d \in \{1,2\}$, and $\beta^d_{\mathtt{P}} = 0.6$ for $d \in \{3,4\}$.
    \item Case 2: $\beta^d_{\mathtt{P}} = 0.6$ for $d \in \{1,2\}$, and $\beta^d_{\mathtt{P}} = 0.1$ for $d \in \{3,4\}$.
\end{itemize}
The values of other parameters are given in Table \ref{tab:tnse1_parameter}.

\begin{figure}
    \centering
    \includegraphics[width=\linewidth]{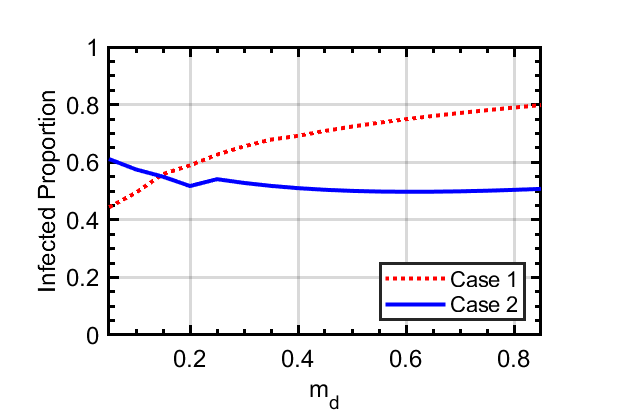}
    \caption{Variation of expected fraction of infected nodes ($y^{\mathtt{avg}}$) at the endemic equilibrium for different values of $m_d$ with $d=4$ and for two possible choice of infection rates stated in Section \ref{section:tnse_casestudy1}.}
    \label{fig:tnse_casestudy1}
\end{figure}

\begin{table}[h]
\begin{center}
\begin{tabular}{|c | c | c | c | c | c | c | } 
 \hline
 $c_{\mathtt{P}}$ & $\alpha$ & $\beta^d_{\mathtt{U}}$ & $\gamma$ & $L$ & $y^d(0)$ & $z^d_{\mathtt{S}}(0)$ \\ [0.5ex] 
 \hline
 15 & 0.5 & 0.6 & 0.2 & 15 & 0.1 & 0.5 \\ \hline
\end{tabular}
 \caption{Values of different parameters used in simulations in Section \ref{section:tnse_casestudy1}.}
 \label{tab:tnse1_parameter}
\end{center}
\end{table}

We induce heterogeneity in the degree distribution by assuming that $m_2 = m_3 = 0.05$, and by varying $m_4$ from $0.05$ to $0.85$. The value of $m_1$ is given by $1-m_2-m_3-m_4$. Figure \ref{fig:tnse_casestudy1} shows the variation of the expected fraction of infected nodes ($y^{\mathtt{avg}}$) at the endemic equilibrium for different values of $m_4$ and for both the cases of infection rates stated above. For the infection rates stated in Case 1, increase in $m_4$ leads to increase $y^{\mathtt{avg}}$ at the endemic equilibrium. This is expected since nodes with higher degree have a larger infection rate, and as the proportion of nodes with degree $d=4$ increases, the overall infection prevalence shows an increase. In contrast for Case 2, increase in $m_4$ leads to increase in the proportion of nodes with smaller $\beta^d_{\mathtt{P}}$, and decrease in the proportion of nodes with larger $\beta^d_{\mathtt{P}}$. As a result, the overall infection prevalence shows a decline, though the decrease in not monotonic. Thus, while intuition may suggest that a larger proportion of high degree nodes would lead to larger infection prevalence, this outcome is not always true when the infection rates are heterogeneous. Rather, an increase in the proportion of nodes with larger infection rates leads to a larger prevalence of the epidemic.  
 
\subsection{Comparison among Degree Distributions}
\label{section:tnse_casestudy_multiple}

In this subsection, we illustrate the impacts of effectiveness of protection $\alpha$, infection probability $\betap^d$ and cost of protection $\cp$ on the infection level at the endemic equilibrium. We let the set of degrees $\DD := \{1,2,\ldots,19,20\}$, and consider three different degree distributions given by
\begin{itemize}
    \item Uniform distribution with $m_d = 0.05$ for all $d \in \DD$,
    \item Binomial distribution with $m_d$ given by the Binomial probability mass function with $n=20$ and $p=0.525$, and
    \item Bimodal distribution with $m_d = 0.25$ for $d \in \{1,2,19,20\}$. 
\end{itemize}
For each of the above degree distributions, the average degree $d^{\mathtt{avg}} = 10.5$. The remaining parameters are set according to Table \ref{tab:tnse1_parameter_multiple}. 

\begin{table}[h]
\begin{center}
\begin{tabular}{|c | c | c | c | c | c | c | } 
 \hline
 $\beta^d_{\mathtt{U}}$ & $\gamma$ & $L$ & $y^d(0)$ & $z^d_{\mathtt{S}}(0)$ \\ [0.5ex] 
 \hline
 0.9 & 0.4 & 10 & 0.1 & 0.5 \\ \hline
\end{tabular}
 \caption{Values of different parameters used in simulations in Section \ref{section:tnse_casestudy_multiple}.}
 \label{tab:tnse1_parameter_multiple}
\end{center}
\end{table}

First we examine the impact of the effectiveness of protection, captured by the parameter $\alpha$. The plots in the left panel of Figure \ref{figure:tnse_degree} shows that as $\alpha$ increases, i.e., the protection becomes less effective, the expected fraction of infected nodes ($y^{\mathtt{avg}}$) as well as $\Theta(\mathbf{y})$ at the endemic equilibrium increase, before saturating for sufficiently large $\alpha$. The plot on the top row also shows that for the entire range of $\alpha$, $y^{\mathtt{avg}}$ under Binomial distribution is larger compared to the network with Uniform distribution followed by the Bimodal distribution. Thus, when the degree distribution of the network is heterogeneous, the expected fraction of infected nodes is smaller compared to a relatively homogeneous degree distribution. However, a similar observation does not hold for $\Theta(\mathbf{y})$. The figure on the bottom row shows that for smaller $\alpha$, $\Theta(\mathbf{y})$ is larger under the Binomial distribution, while for larger $\alpha$, $\Theta(\mathbf{y})$ is larger under the Bimodal distribution. 

Next, we examine the impact of the parameter $\betap^d$ which captures the probability with which an infected protected individual transmits infection to others. The plots in the middle panel of Figure \ref{figure:tnse_degree} shows that as $\betap^d$ increases, both $y^{\mathtt{avg}}$ as well as $\Theta(\mathbf{y})$ at the endemic equilibrium increase. Here also we observe that $y^{\mathtt{avg}}$ is larger under the Binomial distribution, followed by Uniform distribution and Bimodal distribution. However, the value of $\Theta(\mathbf{y})$ is approximately equal for all three degree distributions. 

\begin{figure*}[t]
\centering
  \includegraphics[width=0.33\linewidth]{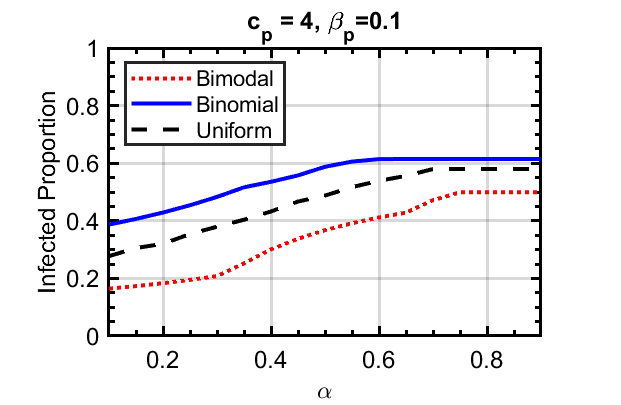}
  \includegraphics[width=0.33\linewidth]{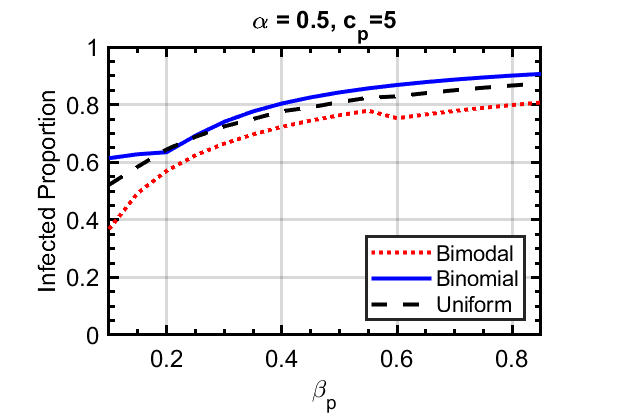}
  \includegraphics[width=0.33\linewidth]{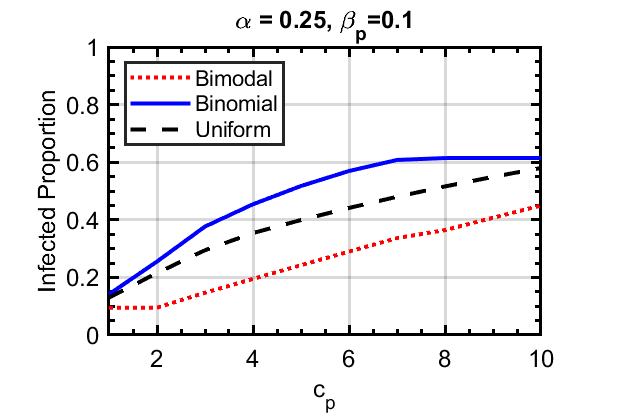}
  \\
  \includegraphics[width=0.33\linewidth]{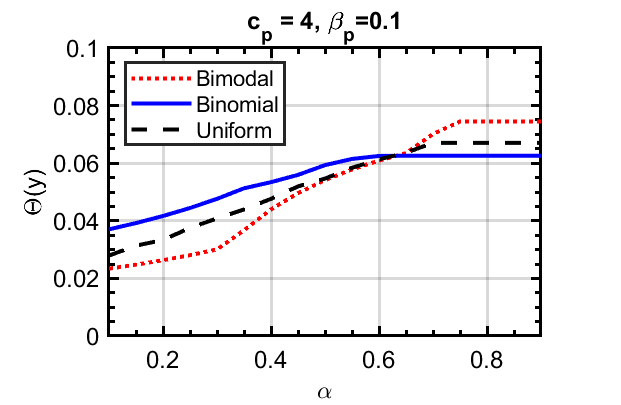}
  \includegraphics[width=0.33\linewidth]{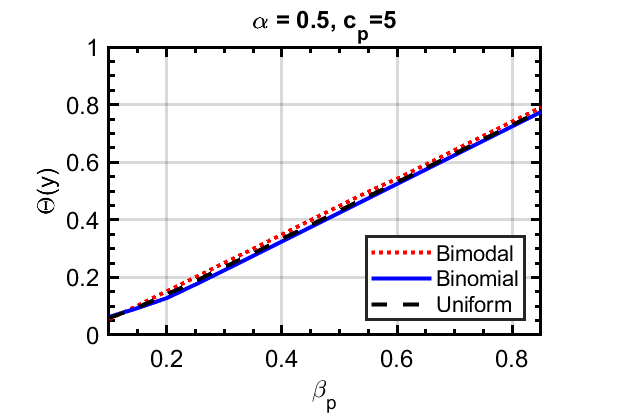}
  \includegraphics[width=0.33\linewidth]{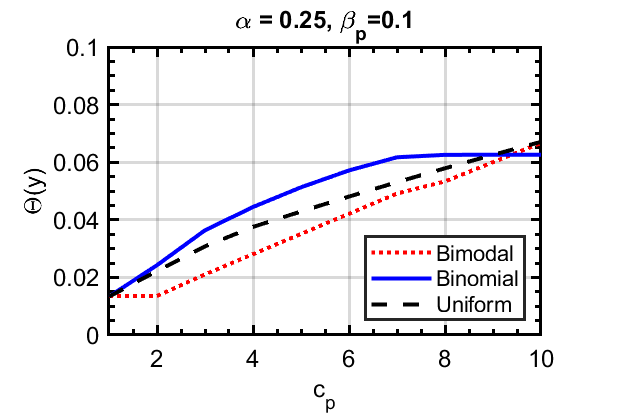}
  \caption{Expected fraction of infected agents (top row) and probability of becoming infected from a randomly chosen neighbor (bottom row) at the equilibrium as a function of effectiveness of protection (left panel), infection probability (middle panel) and cost of protection (right panel). We have assumed $\betap^d=\betap$ for all degrees in this example.}
  \label{figure:tnse_degree}
\end{figure*}

Finally, the plots in the right panel of Figure \ref{figure:tnse_degree} show the impact of the cost of protection adoption $\cp$. As $\cp$ increases, both $y^{\mathtt{avg}}$ as well as $\Theta(\mathbf{y})$ tend to increase. Furthermore, both $y^{\mathtt{avg}}$ as well as $\Theta(\mathbf{y})$ are larger under the Binomial distribution, followed by Uniform and Bimodal distributions, respectively. To summarize, the above numerical results yield the following insights.
\begin{itemize}
    \item When the proportion of nodes with a larger infection rate increases, $y^{\mathtt{avg}}$ at the endemic equilibrium tends to increase.
    \item Similarly, $y^{\mathtt{avg}}$ increases when protection becomes expensive (larger $\cp$) and less effective in both preventing (larger $\alpha$) as well as transmitting (larger $\betap^d$) infection. 
    \item When the degree distribution is nearly homogeneous (e.g., Binomial distribution), $y^{\mathtt{avg}}$ tends to be larger compared to when the degree distribution is largely heterogeneous (e.g., Bimodal and Uniform distributions). 
\end{itemize}

\section{Conclusion}
We analyzed the problem of strategic adoption of partially effective protection in large-scale networks in the population game framework. We derived the coupled epidemic-behavioral dynamics and relied on time-scale separation to derive the epidemic dynamics under optimal protection adoption strategies of the agents which depends on their degree. We then rigorously established the existence and uniqueness of stationary equilibrium of the above dynamics. We numerically illustrated the convergence of the dynamics to the equilibrium as well as the impacts of heterogeneous node degrees, infection rates and cost of protection adoption on the epidemic prevalence at the equilibrium. We aim to leverage the insights derived from this work to design intervention schemes which incentivizes protection adoption among users and reduce the prevalence of epidemics in follow up works. In addition, analyzing the protection adoption behavior of non-myopic or forward-looking agents for networked SIS as well as other classes of epidemic models remain as promising directions for future research. 

%%%%%%%%%%%%%%%%%%%%%%%%%%%%%%%%%%%%%%%%%%%%%%%%%%%
%%%%%%%%%%%%%%%%%%%%%%%%%%%%%%%%%%%%%%%%%%%%%%%%%%%

\appendices

\section{Omitted Proofs}

We first present an important result on the convergence and stability of the networked SIS epidemic under the N-intertwined mean-field approximation (NIMFA) followed by presenting the proofs omitted from the main text. 

\subsection{NIMFA of the SIS Epidemic Model}
\label{section:tnse_appendix_nimfa}

Consider a directed graph or network, denoted $\GG = (\VV,\EE)$ with $\VV$ being the set of nodes and $\EE$ being the set of directed edges. Let $|\VV|=n$, and $A \in \Rb^{n \times n}_{+}$ be the adjacency matrix of the graph. In particular, $a_{ij} = 0$ if and only if $(j,i) \notin \EE$. Let $p_i(t) \in [0,1]$ denote the probability of node $i$ being infected at time $t$. According to the NIMFA of the SIS epidemic \cite{van2008virus}, the infection probability evolves as
\begin{align}\label{eq:app_nimfa}
\frac{d p_i(t)}{dt} & = -\gamma p_i(t) + (1-p_i(t)) \sum^n_{j =1} a_{ij} p_j(t),
\end{align}
where $a_{ij} \geq 0$ denotes the probability of node $i$ becoming infected by node $j$, and $\gamma > 0$ denotes the rate with which an infected node recovers from the disease. The above dynamics can be written in vector form as
\begin{align}\label{eq:app_nimfa_linear}
\dot{p}(t) & = (A-D) p(t) - P(t) A p(t),
\end{align}
where $D = \mathtt{diag}(\gamma,\gamma,\ldots,\gamma)$ is the diagonal matrix of all recovery rates, and $P(t) = \mathtt{diag}(p(t))$. We now reproduce the following theorem from past works regarding the existence and stability of the equilibrium points of \eqref{eq:app_nimfa_linear}. 

\begin{theorem}[\cite{mei2017dynamics,khanafer2016stability}]\label{theorem:nimfa}
    Suppose the graph $\GG$ is strongly connected. Then, 
    \begin{enumerate}
        \item the disease-free equilibrium (DFE) with $p^\star_{\mathtt{DFE}} = 0_n$ is globally asymptotically stable (GAS) if and only if the spectral radius $\rho(D^{-1}A) \leq 1$, and
        \item a unique endemic equilibrium (EE) with $p^\star_{\mathtt{EE}} \gg 0_n$ exists if and only if $\rho(D^{-1}A) > 1$. If $p(0) \neq 0$ and $\rho(D^{-1}A) > 1$, then the endemic equilibrium is GAS. 
    \end{enumerate}
\end{theorem}

\subsection{Proof of Proposition \ref{proposition:tnse_reduced_dynamics}}
\label{section:tnse_appendix_proposition}

\begin{proof}
For the proof, we exploit Theorem \ref{theorem:nimfa} after establishing the equivalence between the dynamics \eqref{eq:sis_switched_2} and the N-Intertwined Mean-Field Approximation (NIMFA) of the SIS epidemic model on a directed network \eqref{eq:app_nimfa}. 

To this end, construct a directed graph $\hat{\GG}$ with $d^{\max}$ nodes, i.e., each degree $d \in \DD$ is treated as a node of $\hat{\GG}$. We define the adjacency matrix $\hat{A}$ where the weight of the edge between two nodes $d$ and $d'$ is given by
\begin{align*}
    [\hat{A}]_{d,d'} := \begin{cases}
    & \frac{d}{d^{\mathtt{avg}}} (d'm_{d'} \betap^{d'}), \qquad \text{for} \quad d < d^\star,
    \\ & \frac{\alpha d}{d^{\mathtt{avg}}} (d'm_{d'} \betap^{d'}), \qquad \text{for} \quad d \geq d^\star.
\end{cases}
\end{align*}

It is now easy to see that the dynamics \eqref{eq:sis_switched_2} is equivalent to the NIMFA approximation of the SIS epidemic \eqref{eq:app_nimfa} on the network $\hat{\GG}$ with adjacency matrix $\hat{A}$. In addition, all entries of $\hat{A}$ are nonzero (due to our assumption that $m_d > 0$ for all $d \in \DD$), and hence, $\hat{\GG}$ is strongly connected. Furthermore, the matrix $D^{-1}\hat{A}$ has rank one since it is the outer product of two vectors given by $D^{-1}\hat{A} = \mathbf{v_1} \cdot \mathbf{v_2}^\top$, where
\begin{align}
     \mathbf{v_1} & = \frac{1}{d^{\mathtt{avg}} \gamma}\begin{bmatrix} 1 \\ 2 \\ \vdots \\ d^\star - 1 \\ \alpha d^\star \\ \vdots \\ \alpha d^{\max}\end{bmatrix}, \quad \mathbf{v_2} = \begin{bmatrix}
        m_{1} \betap^{1} \\ \vdots \\ d'm_{d'} \betap^{d'} \\ \vdots \\ d^{\max}m_{d^{\max}} \betap^{d^{\max}}
    \end{bmatrix}.
\end{align}
As a result, the spectral radius
\begin{align}
\rho(D^{-1}\hat{A}) & = \mathbf{v_1}^\top \mathbf{v_2} = \sum^{d^\star-1}_{d=1} \frac{d^2m_d\betap^{d}}{d^{\mathtt{avg}} \gamma} + \sum^{d^{\max}}_{d = d^\star} \frac{\alpha d^2m_d\betap^{d}}{d^{\mathtt{avg}} \gamma} \nonumber
\\ & =: \RR(d^\star).
\end{align}
The result now follows from Theorem \ref{theorem:nimfa}.
\end{proof}

\subsection{Proof of Theorem \ref{theorem:tnse_main}}
\label{section:tnse_appendix_theorem_proof}

\begin{proof}
{\bf Part 1: $\RR(d^{\max}+1) \leq 1$.} 

\noindent It follows from Lemma \ref{lemma:monotonicity_R_Theta_DBMF} that $\RR(d) < 1$ for all $d \in \{1,2,\ldots,d^{\max}\}$. Assume on the contrary that there exists a nonzero endemic equilibrium denoted $\mathbf{y}^\star_{\mathtt{EE}}$ with $\Theta(\mathbf{y}^{\star}_{\mathtt{EE}}) \in \II_{d'}$ for some $d' \in \{d_{\min},2,\ldots,d^{\max}+1\}$. However, since $\RR(d') < 1$, the disease-free equilibrium is the only equilibrium of the dynamics \eqref{eq:sis_switched_2} with $d^\star=d'$. Since the dynamics \eqref{eq:sis_switched} coincides with \eqref{eq:sis_switched_2} over this interval, there does not exist an equilibrium with $\Theta(\mathbf{y}^{\star}_{\mathtt{EE}}) \in \II_{d'}$ for either \eqref{eq:sis_switched_2} or \eqref{eq:sis_switched}. 

%If $\Theta(\mathbf{y}^{\star}_{\mathtt{EE}}) = \Theta^{d'}_{th}$ for some $d'\in\DD$, then also it follows from identical arguments that there does not exist an endemic equilibrium of \eqref{eq:sis_switched}. 

Now, suppose $\Theta(\mathbf{y}^{\star}_{\mathtt{EE}}) = \Theta^{d'}_{th}$ for some $d'\in\DD$. Then the strategy profile of susceptible agents needs to satisfy
\begin{equation*}
    z^d_{\ST} = \begin{cases}
    0, \qquad & \text{if} \qquad d > d' \quad \text{or} \qquad \Theta^{d}_{th} < \Theta^{d'}_{th},
    \\ z^{d'}_{\ST}, \qquad & \text{if} \qquad d=d'
    \\ 1, \qquad & \text{if} \qquad d < d' \quad \text{or} \qquad \Theta^{d}_{th} > \Theta^{d'}_{th},
    \end{cases}
\end{equation*}
for some $z^{d'}_{\ST} \in [0,1]$. It follows from \eqref{eq:theta_equilibrium_identity} that $\Theta^{d'}_{th}$ satisfies 
\begin{align}
    1 & = \sum_{d \in \DD} \Big[\frac{dm_d}{d^{\mathtt{avg}}} \betap^d \frac{(z^d_{\ST} + \alpha (1-z^d_{\ST})) d }{\gamma + (z^d_{\ST} + \alpha (1-z^d_{\ST})) d \Theta^{d'}_{th}}\Big] \nonumber
    \\ & = \sum^{d'-1}_{d=1} \Big[\frac{dm_d}{d^{\mathtt{avg}}} \frac{d \betap^d}{\gamma + d \Theta^{d'}_{th}}\Big] \nonumber
    \\ & \qquad + \Big[\frac{d'm_{d'}}{d^{\mathtt{avg}}} \frac{\betap^{d'}(z^{d'}_{\ST} + \alpha (1-z^{d'}_{\ST})) d' }{\gamma + (z^{d'}_{\ST} + \alpha (1-z^{d'}_{\ST})) d' \Theta^{d'}_{th}}\Big] \nonumber
    \\ & \qquad + \sum^{d^{\max}}_{d=d'+1} \Big[\frac{dm_d}{d^{\mathtt{avg}}} \frac{\alpha d \betap^d}{\gamma + \alpha d \Theta^{d'}_{th}}\Big] \label{eq:theorem_contradiction1}
    \\ & < \sum^{d^{\max}}_{d=1} \frac{d^2m_d\betap^d}{d^{\mathtt{avg}} (\gamma + \alpha d \Theta^{d'}_{th})} < \sum^{d^{\max}}_{d=1} \frac{d^2m_d\betap^d}{d^{\mathtt{avg}} \gamma} = \RR(d^{\max}+1). \nonumber
\end{align}
The inequality follows because the second term in the R.H.S. of the \eqref{eq:theorem_contradiction1} is monotonically increasing in $z^{d'}_{\ST}$, and the third term is monotonically increasing in $\alpha$. However, this is a contradiction since $\RR(d^{\max}+1) \leq 1$ in this regime. Thus, there does not exist an endemic equilibrium of \eqref{eq:sis_switched} with $\Theta(\mathbf{y}^\star_{\mathtt{EE}}) > 0$.

%%%%%%%%%%%%%%%%%%%%%%%%%%%%%%%%%%%%%%%%%%%%%%%%%%%%%%%%%%%%
%%%%%%%%%%%%%%%%%%%%%%%%%%%%%%%%%%%%%%%%%%%%%%%%%%%%%%%%%%%%

{\bf Part 2: $\RR(d^{\max}+1) > 1$.} 

\noindent Following the definition of $\Theta^d_{th}$ and Lemma \ref{lemma:monotonicity_R_Theta_DBMF}, we have
\begin{align*}
    & 0 = \Theta^{d^{\max}+1}_{th} < \Theta^{d^{\max}}_{th} < \Theta^{d^{\max}-1}_{th} < \ldots < \Theta^{d_{\min}}_{th} < 1 
    \\ & \qquad \qquad \leq \ldots \Theta^{1}_{th},
    \\ & \Theta(\mathbf{y}_{\mathtt{EE}}(1)) \leq \Theta(\mathbf{y}_{\mathtt{EE}}(2)) \leq \ldots \leq \Theta(\mathbf{y}_{\mathtt{EE}}(d^{\max})) 
    \\ & \qquad \qquad \leq \Theta(\mathbf{y}_{\mathtt{EE}}(d^{\max}+1)).
\end{align*}
From the definition of $d^{\mathtt{eq}}$, we have
\begin{align*}
\Theta(\mathbf{y}_{\mathtt{EE}}(d)) & > \Theta^{d}_{th}, \quad \text{for} \quad d \geq d^{\mathtt{eq}} \qquad \text{and} 
\\ \Theta(\mathbf{y}_{\mathtt{EE}}(d)) & \leq \Theta^{d}_{th} \quad \text{for} \quad d < d^{\mathtt{eq}}. 
\end{align*}
We tackle the two cases separately.

\noindent {\bf Case (a):} $\Theta(\mathbf{y}_{\mathtt{EE}}(d^{\mathtt{eq}})) \in (\Theta^{d^{\mathtt{eq}}}_{th},\min(1,\Theta^{d^{\mathtt{eq}}-1}_{th}))$. Note that the dynamics \eqref{eq:sis_switched} coincides with \eqref{eq:sis_switched_2} over the interval $\II_{d^{\mathtt{eq}}}$. Since $\Theta(\mathbf{y}_{\mathtt{EE}}(d^{\mathtt{eq}}))>0$, it is necessarily the case that $\RR(d^{\mathtt{eq}}) > 1$. Consequently, the dynamics \eqref{eq:sis_switched_2} with $d^\star = d^{\mathtt{eq}}$ has a unique nonzero endemic equilibrium at which $\Theta(\mathbf{y}_{\mathtt{EE}}(d^{\mathtt{eq}})) \in \II_{d^{\mathtt{eq}}}$. Therefore, $\mathbf{y}_{\mathtt{EE}}(d^{\mathtt{eq}})$ is a nonzero endemic equilibrium of \eqref{eq:sis_switched}. 

It remains to show that there does not exist any other nonzero endemic equilibrium of \eqref{eq:sis_switched}. Suppose there exists another nonzero endemic equilibrium $\mathbf{y}_{\mathtt{EE},2}$ with $\Theta(\mathbf{y}_{\mathtt{EE},2}) \in \II_{d'}$ for some $d' \neq d^{\mathtt{eq}}$. We examine the following two possibilities.
\begin{itemize}
    \item Suppose $d' > d^{\mathtt{eq}}$. The dynamics \eqref{eq:sis_switched_2} with $d^\star = d' > d^{\mathtt{eq}}$ has a unique nonzero endemic equilibrium, and following Lemma \ref{lemma:monotonicity_R_Theta_DBMF}, we have $\Theta(\mathbf{y}_{\mathtt{EE}}(d')) > \Theta(\mathbf{y}_{\mathtt{EE}}(d^{\mathtt{eq}})) > \Theta^{d^{\mathtt{eq}}}_{th}$. However, $\Theta(\mathbf{y}_{\mathtt{EE}}(d')) \notin \II_{d'}$ because $\II_{d'} = (\Theta^{d'}_{th},\Theta^{d'-1}_{th})$, and $\Theta^{d'-1}_{th} \leq \Theta^{d^{\mathtt{eq}}}_{th}$. As a result, $\mathbf{y}_{\mathtt{EE}}(d')$ is not an endemic equilibrium for \eqref{eq:sis_switched}.
    \item Now, let $d' < d^{\mathtt{eq}}$. The dynamics \eqref{eq:sis_switched_2} with $d^\star = d' < d^{\mathtt{eq}}$ has a unique nonzero endemic equilibrium, and following the definition of $d^{\mathtt{eq}}$, we have $\Theta(\mathbf{y}_{\mathtt{EE}}(d')) < \Theta^{d'}_{th}$. As a result, $\Theta(\mathbf{y}_{\mathtt{EE}}(d')) \notin \II_{d'} = (\Theta^{d'}_{th},\Theta^{d'-1}_{th})$. Thus, $\mathbf{y}_{\mathtt{EE}}(d')$ is not an endemic equilibrium for \eqref{eq:sis_switched}.
\end{itemize}

Now, suppose there exists another nonzero endemic equilibrium $\mathbf{y}_{\mathtt{EE},2}$ such that $\Theta(\mathbf{y}_{\mathtt{EE},2}) = \Theta^{d'}_{th}$ for some $d'$. Then, $\Theta^{d'}_{th}$ satisfies \eqref{eq:theorem_contradiction1}. Let $d' \geq d^{\mathtt{eq}}$. Since $\Theta(\mathbf{y}_{\mathtt{EE}}(d^{\mathtt{eq}})) \in (\Theta^{d^{\mathtt{eq}}}_{th},\min(1,\Theta^{d^{\mathtt{eq}}-1}_{th}))$, we have $\Theta^{d'}_{th} \leq \Theta^{d^{\mathtt{eq}}}_{th} < \Theta(\mathbf{y}_{\mathtt{EE}}(d^{\mathtt{eq}}))$. Setting $d^\star = d^{\mathtt{eq}}$ in \eqref{eq:theta_star_d_star}, we obtain
\begin{align*}
    1 & = \sum^{d^{\mathtt{eq}}-1}_{d =1} \Big[\frac{dm_d}{d^{\mathtt{avg}}} \frac{d \betap^d}{\gamma + d \Theta(\mathbf{y}_{\mathtt{EE}}(d^{\mathtt{eq}}))}\Big] 
    \\ & \qquad + \sum^{d^{\max}}_{d =d^{\mathtt{eq}}} \Big[\frac{dm_d}{d^{\mathtt{avg}}} \frac{\alpha d \betap^d}{\gamma + \alpha d \Theta(\mathbf{y}_{\mathtt{EE}}(d^{\mathtt{eq}}))}\Big]
    \\ & \leq \sum^{d^{\mathtt{eq}}-1}_{d =1} \Big[\frac{dm_d}{d^{\mathtt{avg}}} \frac{d \betap^d}{\gamma + d \Theta^{d'}_{th}}\Big] + \sum^{d^{\max}}_{d =d^{\mathtt{eq}}} \Big[\frac{dm_d}{d^{\mathtt{avg}}} \frac{\alpha d \betap^d}{\gamma + \alpha d \Theta^{d'}_{th}}\Big]
    \\ & = \sum^{d^{\mathtt{eq}}-1}_{d =1} \Big[\frac{dm_d}{d^{\mathtt{avg}}} \frac{d \betap^d}{\gamma + d \Theta^{d'}_{th}}\Big] + \! \! \sum^{d'-1}_{d =d^{\mathtt{eq}}} \Big[\frac{dm_d}{d^{\mathtt{avg}}} \frac{\alpha d \betap^d}{\gamma + \alpha d \Theta^{d'}_{th}}\Big] 
    \\ & \qquad +  \frac{d'm_{d'}}{d^{\mathtt{avg}}} \frac{\alpha d' \betap^{d'}}{\gamma + \alpha d' \Theta^{d'}_{th}} + \! \! \sum^{d^{\max}}_{d =d'+1} \Big[\frac{dm_d}{d^{\mathtt{avg}}} \frac{\alpha d \betap^d}{\gamma + \alpha d \Theta^{d'}_{th}}\Big]
    \\ & < \sum^{d^{'}-1}_{d =1} \Big[\frac{dm_d}{d^{\mathtt{avg}}} \frac{d \betap^d}{\gamma + d \Theta^{d'}_{th}}\Big] 
    \\ & \qquad + \Big[\frac{d'm_{d'}}{d^{\mathtt{avg}}} \betap^{d'}  \frac{(z^{d'}_{\ST} + \alpha (1-z^{d'}_{\ST})) d' }{\gamma + (z^{d'}_{\ST} + \alpha (1-z^{d'}_{\ST})) d' \Theta^{d'}_{th}}\Big] 
    \\ & \qquad + \sum^{d^{\max}}_{d =d'+1} \Big[\frac{dm_d}{d^{\mathtt{avg}}} \frac{\alpha d \betap^d}{\gamma + \alpha d \Theta^{d'}_{th}}\Big]
\end{align*}
for any $z^{d'}_{\ST} \in [0,1]$ due to the monotonicity of the second term in $z^{d'}_{\ST}$ and $\alpha$. However, this is in contradiction to \eqref{eq:theorem_contradiction1} which requires the R.H.S. to be equal to $1$. 

Now suppose $d' < d^{\mathtt{eq}}$. Since $\Theta(\mathbf{y}_{\mathtt{EE}}(d^{\mathtt{eq}})) \in (\Theta^{d^{\mathtt{eq}}}_{th},\min(1,\Theta^{d^{\mathtt{eq}}-1}_{th}))$, we have $\Theta(\mathbf{y}_{\mathtt{EE}}(d^{\mathtt{eq}})) < \Theta^{d'}_{th}$. Proceeding as before, we obtain
\begin{align*}
    1 & = \sum^{d^{\mathtt{eq}}-1}_{d =1} \Big[\frac{dm_d}{d^{\mathtt{avg}}} \frac{d \betap^d}{\gamma + d \Theta(\mathbf{y}_{\mathtt{EE}}(d^{\mathtt{eq}}))}\Big] 
    \\ & \qquad + \sum^{d^{\max}}_{d =d^{\mathtt{eq}}} \Big[\frac{dm_d}{d^{\mathtt{avg}}} \frac{\alpha d \betap^d}{\gamma + \alpha d \Theta(\mathbf{y}_{\mathtt{EE}}(d^{\mathtt{eq}}))}\Big]
    \\ & > \sum^{d^{\mathtt{eq}}-1}_{d =1} \Big[\frac{dm_d}{d^{\mathtt{avg}}} \frac{d \betap^d}{\gamma + d \Theta^{d'}_{th}}\Big] + \sum^{d^{\max}}_{d =d^{\mathtt{eq}}} \Big[\frac{dm_d}{d^{\mathtt{avg}}} \frac{\alpha d \betap^d}{\gamma + \alpha d \Theta^{d'}_{th}}\Big]
    \\ & = \sum^{d^{'}-1}_{d =1} \Big[\frac{dm_d}{d^{\mathtt{avg}}} \frac{d \betap^d}{\gamma + d \Theta^{d'}_{th}}\Big] + \frac{d'm_{d'}}{d^{\mathtt{avg}}} \frac{d' \betap^{d'}}{\gamma + d' \Theta^{d'}_{th}} 
    \\ & \qquad + \sum^{d^{\mathtt{eq}}-1}_{d =d'+1} \Big[\frac{dm_d}{d^{\mathtt{avg}}} \frac{d \betap^d}{\gamma + d \Theta^{d'}_{th}}\Big] 
    \\ & \qquad + \sum^{d^{\max}}_{d =d^{\mathtt{eq}}} \Big[\frac{dm_d}{d^{\mathtt{avg}}} \frac{\alpha d \betap^d}{\gamma + \alpha d \Theta^{d'}_{th}}\Big]
    \\ & \geq \sum^{d^{'}-1}_{d =1} \Big[\frac{dm_d}{d^{\mathtt{avg}}} \frac{d \betap^d}{\gamma + d \Theta^{d'}_{th}}\Big] + \frac{d'm_{d'}}{d^{\mathtt{avg}}} \frac{d' \betap^{d'}}{\gamma + d' \Theta^{d'}_{th}} 
    \\ & \qquad + \sum^{d^{\max}}_{d =d^{'}+1} \Big[\frac{dm_d}{d^{\mathtt{avg}}} \frac{\alpha d \betap^d}{\gamma + \alpha d \Theta^{d'}_{th}}\Big]
    \\ & \geq \sum^{d^{'}-1}_{d =1} \Big[\frac{dm_d}{d^{\mathtt{avg}}} \frac{d \betap^d}{\gamma + d \Theta^{d'}_{th}}\Big] 
    \\ & \qquad + \Big[\frac{d'm_{d'}}{d^{\mathtt{avg}}} \betap^{d'}  \frac{(z^{d'}_{\ST} + \alpha (1-z^{d'}_{\ST})) d' }{\gamma + (z^{d'}_{\ST} + \alpha (1-z^{d'}_{\ST})) d' \Theta^{d'}_{th}}\Big] 
    \\ & \qquad + \sum^{d^{\max}}_{d =d'+1} \Big[\frac{dm_d}{d^{\mathtt{avg}}} \frac{\alpha d \betap^d}{\gamma + \alpha d \Theta^{d'}_{th}}\Big]
\end{align*}
for any $z^{d'}_{\ST} \in [0,1]$. As before, this is in contradiction to \eqref{eq:theorem_contradiction1}. Therefore, $\mathbf{y}_{\mathtt{EE}}(d^{\mathtt{eq}})$ is the unique endemic equilibrium of \eqref{eq:sis_switched} in this regime. 

%%%%%%%%%%%%%%%%%%%%%%%%%%%%%%%%%%%%%%%%%%%%%%%%%%%%%%%%%%%%
%%%%%%%%%%%%%%%%%%%%%%%%%%%%%%%%%%%%%%%%%%%%%%%%%%%%%%%%%%%%

\noindent {\bf Case (b):} $\Theta(\mathbf{y}_{\mathtt{EE}}(d^{\mathtt{eq}})) \geq \Theta^{d^{\mathtt{eq}}-1}_{th}$. It follows from the definition of $d^{\mathtt{eq}}$ that $\Theta(\mathbf{y}_{\mathtt{EE}}(d^{\mathtt{eq}}-1)) \leq \Theta^{d^{\mathtt{eq}}-1}_{th}$. Note further than both $\Theta(\mathbf{y}_{\mathtt{EE}}(d^{\mathtt{eq}}))$ and $\Theta(\mathbf{y}_{\mathtt{EE}}(d^{\mathtt{eq}}-1))$ can not simultaneously be equal to $\Theta^{d^{\mathtt{eq}}-1}_{th}$ due to the strict monotonicity established in Lemma \ref{lemma:monotonicity_R_Theta_DBMF}. By setting $d^\star = d^{\mathtt{eq}}$ and $d^\star = d^{\mathtt{eq}}-1$ in \eqref{eq:theta_star_d_star}, we respectively obtain
\begin{align*}
1 & = \sum^{d^{\mathtt{eq}}-2}_{d =1} \Big[\frac{dm_d}{d^{\mathtt{avg}}} \frac{d \betap^d}{\gamma + d \Theta(\mathbf{y}_{\mathtt{EE}}(d^{\mathtt{eq}}))}\Big] 
\\ & \qquad + \sum^{d^{\max}}_{d =d^{\mathtt{eq}}} \Big[\frac{dm_d}{d^{\mathtt{avg}}} \frac{\alpha d \betap^d}{\gamma + \alpha d \Theta(\mathbf{y}_{\mathtt{EE}}(d^{\mathtt{eq}}))}\Big] 
\\ & \qquad + \frac{(d^{\mathtt{eq}}-1)m_{d^{\mathtt{eq}}-1}}{d^{\mathtt{avg}}} \frac{(d^{\mathtt{eq}}-1) \betap^{d^{\mathtt{eq}}-1}}{\gamma + (d^{\mathtt{eq}}-1) \Theta(\mathbf{y}_{\mathtt{EE}}(d^{\mathtt{eq}}))}, 
\\ & \leq \sum^{d^{\mathtt{eq}}-2}_{d =1} \Big[\frac{dm_d}{d^{\mathtt{avg}}} \frac{d \betap^d}{\gamma + d \Theta^{d^{\mathtt{eq}}-1}_{th}}\Big] + \sum^{d^{\max}}_{d =d^{\mathtt{eq}}} \Big[\frac{dm_d}{d^{\mathtt{avg}}} \frac{\alpha d \betap^d}{\gamma + \alpha d \Theta^{d^{\mathtt{eq}}-1}_{th}}\Big] 
\\ & \qquad + \frac{(d^{\mathtt{eq}}-1)m_{d^{\mathtt{eq}}-1}}{d^{\mathtt{avg}}} \frac{(d^{\mathtt{eq}}-1) \betap^{d^{\mathtt{eq}}-1}}{\gamma + (d^{\mathtt{eq}}-1) \Theta^{d^{\mathtt{eq}}-1}_{th}}, 
\\ 1 & = \sum^{d^{\mathtt{eq}}-2}_{d =1} \Big[\frac{dm_d}{d^{\mathtt{avg}}} \frac{d \betap^d}{\gamma + d \Theta(\mathbf{y}_{\mathtt{EE}}(d^{\mathtt{eq}}-1))}\Big] 
\\ & \qquad + \sum^{d^{\max}}_{d =d^{\mathtt{eq}}} \Big[\frac{dm_d}{d^{\mathtt{avg}}} \frac{\alpha d \betap^d}{\gamma + \alpha d \Theta(\mathbf{y}_{\mathtt{EE}}(d^{\mathtt{eq}}-1))}\Big] 
\\ & \qquad + \frac{(d^{\mathtt{eq}}-1)m_{d^{\mathtt{eq}}-1}}{d^{\mathtt{avg}}} \frac{\alpha (d^{\mathtt{eq}}-1) \betap^{d^{\mathtt{eq}}-1}}{\gamma + \alpha (d^{\mathtt{eq}}-1) \Theta(\mathbf{y}_{\mathtt{EE}}(d^{\mathtt{eq}}-1))}
\\ & \geq \sum^{d^{\mathtt{eq}}-2}_{d =1} \Big[\frac{dm_d}{d^{\mathtt{avg}}} \frac{d \betap^d}{\gamma + d \Theta^{d^{\mathtt{eq}}-1}_{th}}\Big] + \sum^{d^{\max}}_{d =d^{\mathtt{eq}}} \Big[\frac{dm_d}{d^{\mathtt{avg}}} \frac{\alpha d \betap^d}{\gamma + \alpha d \Theta^{d^{\mathtt{eq}}-1}_{th}}\Big] 
\\ & \qquad + \frac{ (d^{\mathtt{eq}}-1)m_{d^{\mathtt{eq}}-1}}{d^{\mathtt{avg}}} \frac{\alpha (d^{\mathtt{eq}}-1) \betap^{d^{\mathtt{eq}}-1}}{\gamma + \alpha (d^{\mathtt{eq}}-1) \Theta^{d^{\mathtt{eq}}-1}_{th}}.
\end{align*}

We now define
\begin{align}
\bar{\RR} & := 1 -  \sum^{d^{\mathtt{eq}}-2}_{d =1} \Big[\frac{dm_d}{d^{\mathtt{avg}}} \frac{d \betap^d}{\gamma + d \Theta^{d^{\mathtt{eq}}-1}_{th}}\Big] \nonumber
\\ & \qquad - \sum^{d^{\max}}_{d =d^{\mathtt{eq}}} \Big[\frac{dm_d}{d^{\mathtt{avg}}} \frac{\alpha d \betap^d}{\gamma + \alpha d \Theta^{d^{\mathtt{eq}}-1}_{th}}\Big]
\\ \implies & \frac{ (d^{\mathtt{eq}}-1)m_{d^{\mathtt{eq}}-1}}{d^{\mathtt{avg}}} \frac{\alpha (d^{\mathtt{eq}}-1) \betap^{d^{\mathtt{eq}}-1}}{\gamma + \alpha (d^{\mathtt{eq}}-1) \Theta^{d^{\mathtt{eq}}-1}_{th}} \leq \bar{\RR} \nonumber
\\ & \qquad \leq \frac{(d^{\mathtt{eq}}-1)m_{d^{\mathtt{eq}}-1}}{d^{\mathtt{avg}}} \frac{(d^{\mathtt{eq}}-1) \betap^{d^{\mathtt{eq}}-1}}{\gamma + (d^{\mathtt{eq}}-1) \Theta^{d^{\mathtt{eq}}-1}_{th}}.
\end{align}
Consequently, there exists a unique $\bar{z}^{d^{\mathtt{eq}}-1}_{\ST} \in [0,1]$ at which
\begin{align*}
    \bar{\RR} & = \frac{(d^{\mathtt{eq}}-1)m_{d^{\mathtt{eq}}-1}}{d^{\mathtt{avg}}} \times 
    \\ & \qquad \frac{(d^{\mathtt{eq}}-1) (\bar{z}^{d^{\mathtt{eq}}-1}_{\ST}+\alpha(1-\bar{z}^{d^{\mathtt{eq}}-1}_{\ST})) \betap^{d^{\mathtt{eq}}-1}}{\gamma + (\bar{z}^{d^{\mathtt{eq}}-1}_{\ST}+\alpha(1-\bar{z}^{d^{\mathtt{eq}}-1}_{\ST})) (d^{\mathtt{eq}}-1) \Theta^{d^{\mathtt{eq}}-1}_{th}}.
\end{align*}

Now, consider the strategy profile of susceptible agents given by 
\begin{equation*}
    z^d_{\ST} = \begin{cases}
    0, \qquad & \text{if} \qquad d > d^{\mathtt{eq}}-1,
    \\ \bar{z}^{d^{\mathtt{eq}}-1}_{\ST}, \qquad & \text{if} \qquad d=d^{\mathtt{eq}}-1
    \\ 1, \qquad & \text{if} \qquad d < d^{\mathtt{eq}}-1.
    \end{cases}
\end{equation*}
Let $z^d_{\IT} = 0$ for all $d \in \DD$. It is easy to verify that $\Theta^{d^{\mathtt{eq}}-1}_{th}$ is a nonzero solution of \eqref{eq:theta_equilibrium_identity} under the above strategy profile. Consequently, the set of infected proportions $\{y^d\}_{d \in \DD}$ satisfying \eqref{eq:yd_equilibrium} with $\Theta^{d^{\mathtt{eq}}-1}_{th}$ constitutes an equilibrium of the dynamics \eqref{eq:sis_switched}.

The uniqueness of the endemic equilibrium can be established in a manner analogous to the similar uniqueness result established for {\bf Case (a)} above, and is omitted in the interest of space.
\end{proof} 
%%%%%%%%%%%%%%%%%%%%%%%%%%%%%%%%%%%%%%%%%%%%%%%%%%%%%%%%%%%%%%%%%%%%%%%%%%

\bibliographystyle{ieeetr}
\bibliography{main}

\end{document}